\documentclass[11pt,reqno]{amsart}
\usepackage[samelinetheorem,notitle]{maherart}

\usepackage{skak}

\hypersetup{
    pdftitle =
        {Choice Consistency with Learning from Common Information},
    pdfauthor =
        {Rahul Deb and Ludovic Renou},
    pdfsubject=
        {learning, revealed preference, voting, information}
}

\newcommand{\mI}{\mathcal{I}}
\newcommand{\mT}{\mathcal{T}}
\newcommand{\mX}{\mathcal{X}}

\begin{document}

\title{Dynamic Choice and Common Learning}

\author[Deb]{Rahul Deb$^{\between}$}
\address{$^{\between}$Department of Economics, University of Toronto\\\href{mailto:rahul.deb@utoronto.ca}{rahul.deb@utoronto.ca}}
\author[Renou]{Ludovic Renou$^{\dag}$ \\ \today}
\address{$^{\dag}$School of Economics and Finance, Queen Mary University of London\\\href{mailto: lrenou.econ@gmail.com}{lrenou.econ@gmail.com}}

\thanks{This paper was conceptualized at the Studienzentrum Gerzensee and we are grateful for their gracious hospitality. We thank Laura Doval, Alex Jakobsen, Paola Manzini and several conference, seminar participants for comments, suggestions and criticisms that greatly improved the paper. Deb is always grateful to the SSHRC for their continued and generous financial support. Renou gratefully acknowledges the support of the Agence Nationale pour la Recherche under grant ANR CIGNE (ANR-15-CE38-0007-01) and through the ORA Project ``Ambiguity in Dynamic Environments'' (ANR-18-ORAR-0005).}

\begin{abstract}
A researcher observes a finite sequence of choices made by multiple agents in a binary-state environment. Agents maximize expected utilities that depend on their chosen alternative and the unknown underlying state. Agents learn about the time-varying state from the same information and their actions change because of the evolving common belief. The researcher does not observe  agents' preferences, the prior, the common information and the stochastic process for the state. We characterize the set of choices that are rationalized by this model and generalize the information environments to allow for private information. We discuss the implications of our results for uncovering  discrimination and committee decision making.
\end{abstract}

\maketitle

\section{Introduction}
Sets of individuals frequently make decisions based on the same information. For example, in court cases, judges are presented with the same evidence and arguments prior to reaching their decisions. Employers and graduate schools peruse the same curriculum vitae and reference letters prior to inviting applicants for interviews. As yet another instance, deliberative committees such as monetary policy committees gather the opinions and views of all members prior to making decisions. Naturally, even though individuals have the same information, they may still make distinct decisions because their preferences differ. But different preferences alone may not be able to rationalize \textit{all} patterns of choices. What are the rationalizable patterns of choices? Can individual choices reveal changing preferences? If so, can this information be used, for instance, to test for discrimination? \medskip

To address these questions as well as others, we consider a simple binary-state environment in which multiple agents make choices repeatedly. A researcher observes the choices over time, but observes neither the preferences of the agents nor their information.  The researcher assumes that the agents agree about the evolution of the state and learn about the state from common signals. We fully characterize the testable implications of this model, that is, we derive necessary and sufficient conditions for choice data to be rationalizable.\medskip 

We begin with a baseline model in which there are two alternatives and all agents get a higher utility from choosing the alternative that matches the state. The characterization (\cref{thm:baseline}) of this model is easy to state: choice data can be rationalized if, and only if, there are no two periods and two agents who pick different and opposite alternatives in each period. This shows that the model can be refuted with just two periods of choice data despite the researcher not observing the state or its evolution, the preferences, the prior belief and the common signals that agents use to learn. Moreover, the proof of this result demonstrates that we can extract preference rankings across agents. As we discuss later, this information can be useful for applications.\medskip
 
We then generalize the baseline model along several dimensions. We first demonstrate how to accommodate general preferences; agents may get a higher utility by choosing the alternative that mismatches the state (\cref{cor:gen_prefs}). Then, we allow for multiple actions (\cref{thm:mult_actions}). Finally, we relax the assumption of common signals and instead assume that agents' information, while different, is related. Specifically agents observe private signals that differ in their informativeness but favor the same state. In other words, agents agree in each period about whether the information is `good' or `bad' news but disagree about how conclusive the new evidence is. We characterize the testable implications of this general model when the state is time-invariant (\cref{thm:comon_inv}) and show that all choice data can be rationalized if the state is time varying (\cref{thm:comon_var}). \medskip 

To illustrate the range of possible applications, it is instructive to consider a simple example. There are two employers and six applications from prospective employees. Applications $m_1$, $m_2$ and $m_3$, are from male applicants, while applications $f_1$, $f_2$ and $f_3$ are from female applicants. The productivity of each applicant is either low or high (the state).  A researcher observes which applicants an employer invites for interviews, but does not observe the curricula vitae and the reference letters -- the common signals that the employers have. \cref{table:intro} summarizes the researcher's choice data.  

\begin{table}[h]
\caption{Employer's decisions: yes = ``invite for an interview.''}\label{table:intro}
\begin{tabular}{|c|c|c|c|c|c|c|} \hline
& $m_1$ & $m_2$ & $m_3$ & $f_1$ & $f_2$ & $f_3$ \\  \hline
Employer 1 & yes & yes  & no & no & no & yes \\ \hline  
Employer 2 & yes & no   & no & no & yes & yes \\ \hline
\end{tabular}
\end{table}
\medskip 

As in our baseline model, suppose employers only want to interview high-productivity applicants. Both employers may differ, however, in their disutility of type I and type II errors. We now explain how to use our characterization to test for taste-based discrimination. The key observation is that we can not only test whether the entire sample is rationalizable, but also whether the two sub-samples of male and female applicants are. It is immediate to check that the entire dataset is not rationalizable as there are two different candidates, $m_2$ and $f_2$, with opposite and distinct interview decisions by both employers. Thus, we reject the hypothesis that interview decisions are made based on common beliefs and employer-specific preferences that are identical over male and female candidates.

Yet, the two sub-samples are both individually rationalizable. Thus, we cannot reject the hypothesis that the employers have common beliefs but difference preferences over male and female candidates. In fact, we can even make further inferences about the preferences of the employers. Note that whenever Employer 2 invites a male applicant, so does Employer 1. Since with a binary state, an employer invites an applicant only when the common belief is above a threshold, we can infer that Employer 2's threshold for male applicants is  strictly above that of Employer 1. (The threshold is the belief at which an employer is indifferent between the two possible choices.) The opposite is true on the sub-sample of female applicants. \medskip 

A simple application of our tests thus provides suggestive evidence of taste-based discrimination without requiring any information about the quality of the male and female applicants. We do not even need the two sub-samples to be balanced or nearly identical in terms of applications. Also note that there is nothing special about the context of this example; we could replace employers by judges, job applicants by defendants, gender by race and the interview invitations by judicial decisions. In \cref{sec:applications}, we discuss in greater detail the application of our tests to the issue of discrimination, along with other applications.

\subsection*{Related Literature}

This paper is related to the recent decision theory literature that aims to characterize different models of learning and costly information acquisition. A few recent examples are \citet{caplin2015}, \citet{lu2016} and \citet{frick2019}. Perhaps the main difference of our work with this literature is the data on which our respective analyses are conducted. Decision theory papers typically assume that the analyst can observe agents' choices in \textit{every} possible choice scenario. In the case of random choice induced by learning/information acquisition, this amounts to assuming that the analyst can observe the choice \textit{probabilities} of the various alternatives in all choice problems. Needless to say, this is a demanding data requirement.\medskip 

By contrast, we only assume that finitely many \textit{realized} choices are observed from a single choice problem and that the environment under which these choices are made is changing. While the data requirement for our analysis is modest relative to the above literature, the tradeoff is that it is not possible to identify the information (and identification is often an important goal in a decision theoretic analysis). Our results are driven by the assumption that we observe the choices of multiple individuals whose beliefs are systematically related. In practice, versions of this assumption are routinely made implicitly or explicitly.\footnote{Outside of decision theory, this assumption is frequently made in the finance and information design literatures in which information often takes the form of public signals.} Choice probabilities (as opposed to realized choices) are not observed and must be estimated. It is quite unusual to be able to estimate choice probabilities from a single individual as this would require many observations (for any meaningful precision) from the identical choice problem (the same set of alternatives in an identical information environment). So instead, choice probabilities can be obtained by averaging across the choices of different (typically observationally equivalent) individuals in the data. This implicitly assumes that  (conditioning on covariates) these scenarios are treated as identical or, in other words, that different agents have the same preferences and information (but may get different signal realizations).\medskip

The most closely related papers are \citet{shmaya2016} and \citet{deoliveira2019} which also study the testable implications of learning on observed choices. \citet{shmaya2016} study a single agent decision problem in which an agent observes signals and has to report which state she believes to be the most likely. They show that, absent additional structure on information, all choices can be explained by Bayesian updating. \citet{deoliveira2019} consider a related environment but they allow for completely general (but known to the researcher) dynamic preferences that need not be time separable. They show that not all choices are possible and derive a dominance condition that characterizes rationalizability. \cite{makris2021} generalize the concept of Bayes correlated equilibrium to multi-stage games. As an application of their results, they reformulate the problem of \citet{deoliveira2019} as characterizing the Bayes' correlated equilibria of a dynamic decision problem.\medskip

The main difference between these papers and ours is that we do not assume that the researcher knows the agents' preferences. This is an important generalization because preferences are unobserved in all field data. Instead, our main assumption that generates testable predictions is the fact that different agents' information is related. We view this modeling assumption to be one of the main conceptual novelties of our framework.\medskip 

\section{The baseline model}\label{sec:model}

This section presents and discusses the baseline model, which we analyze first. \cref{sec:generalizations} generalizes it in several directions.

\medskip

A set of $\mI:=\{1,\dots,I\}$ agents makes choices over $\mT=\{1,\dots,T\}$ periods. In each period $t$,  agent $i$ chooses an \textit{action} $x_{i,t}\in X_t$ from a binary set $X_t=\{x,y\}$ for all $t$. A researcher observes \textit{choice data} $\mX:=(x_{i,t})_{(i,t) \in \mI \times \mT}$ that contains the choices of all agents over all time periods. This is all the researcher observes. We explain later how to deal with richer choice data, where the researcher observes additional characteristics of the choice problems.   

\medskip

We now define the model the researcher postulates to rationalize the choice data $\mX$. There is an unknown, possibly time-varying, underlying \textit{state} of the world $\omega_t\in \{\omega^x,\omega^y\}$ at each period $t$. We deliberately index the state by the alternative to indicate that action $x$ (resp., $y$) is the best in state $\omega^x$ (resp., $\omega^y$). The researcher thus assumes that agent $i$'s utility function $u_i$  satisfies $u_i(x,\omega^x)\geq u_i(y,\omega^x)$ and $u_i(y,\omega^y) \geq u_i(x,\omega^y)$, with at least one strict inequality. For any utility function $u_i$, we denote
\begin{align*}
	 \overline{u}_i:=\frac{u_i(y,\omega^y)-u_i(x,\omega^y)}{u_i(x,\omega^x) -u_i(y,\omega^x)+ u_i(y,\omega^y) - u_i(x,\omega^y)}.
\end{align*}
the belief threshold at which the agent is indifferent between choosing $x$ and $y$. 

\smallskip

We now turn to the beliefs and the information. The researcher assumes that  the agents share a common \textit{prior} $p_0$ and agree on the evolution of the state. The \textit{transition probability} from state $\omega_{t-1}$ to $\omega_{t}$ is  $\gamma_t(\omega_{t}\,|\,\omega_{t-1})$. We say the state is \textit{time-invariant} if $\gamma_t(\omega^x\,|\,\omega^x)=\gamma_t(\omega^y\,|\,\omega^y)=1$ for all $t\in\mathcal{T}$; otherwise, we say the state is \textit{time-varying}.

\medskip

Prior to choosing an action (but after the state has evolved), the agents observe a \textit{common signal} $s_t$, a realization of the statistical \textit{experiment}  $\pi_t: \{\omega^x,\omega^y\} \rightarrow \Delta(S_t)$, with $S_t$ a finite set of signals.\medskip 

Finally, the researcher assumes that the agents are expected utility maximizers and update their beliefs by Bayes' rule. We use $p_t$ to denote the \textit{belief} that the state is $\omega^x$ at period $t$ (the probability of $\omega^y$ is the complementary probability). The common belief at period $t$ is thus: 
\begin{align}
	p_{t}&=\frac{	\left[p_{t-1}\gamma_t(\omega^x|\omega^x)+(1-p_{t-1})\gamma_t(\omega^x|\omega^y)\right]\pi_t(s_t|\omega^x)}{\sum_{\omega\in\{\omega^x,\omega^y\}}\left[p_{t-1}\gamma_t(\omega|\omega^x)+(1-p_{t-1})\gamma_t(\omega|\omega^y)\right]\pi_t(s_t|\omega)}.\label{eq:Bayes}
\end{align}
Hence, agent $i$ chooses $x$ at period $t$ if and only if $p_t \geq \overline{u}_i$. 

\medskip

For ease of reference, the timing is summarized below

\begin{figure*}[h]
	\hspace*{-1cm}
	\includegraphics[scale=.8]{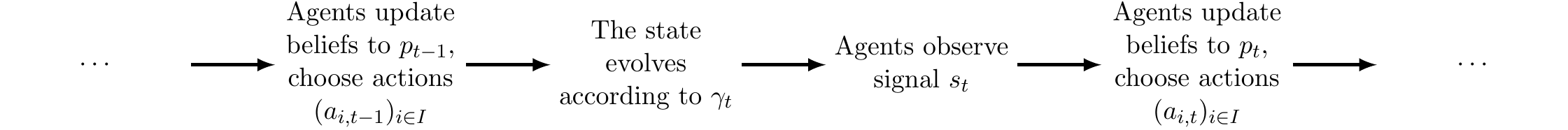}
\end{figure*}

\medskip

The aim of this paper is characterize choice data that can be rationalized by this model. In words, we are interested in when a researcher who only observes choice data $\mX$ can find some parameters for the model such that, given those parameters, the agents would choose precisely the observed data. Formally, this is defined as follows.

\begin{definition}\label{def:rat}
	The choice data $\mathcal{X}$ is \textbf{rationalizable} if, for all $i\in\mI$ and $t\in\mT$, there exist utility functions $u_i$, transition probabilities $\gamma_t$, prior $p_{0}$, beliefs $p_{t}$, sets of signals $S_t$, experiments $\pi_t: \{\omega^x,\omega^y\} \rightarrow \Delta(S_t)$ and realized signals $s_t \in S_t$ such that 
	\begin{enumerate}[label=\roman*.]		
		\item \textit{Actions maximize utility:} $p_{t} >  \overline{u}_i$ if $x_{i,t}=x$ (resp., $p_{t} <  \overline{u}_i$ if $x_{i,t}=y$),
		\item \textit{Bayesian Updating:} $p_{t}$ is derived from $p_{t-1}$ by Bayes' rule.
	\end{enumerate}
\end{definition}
\noindent In what follows, we will sometimes explicitly add a qualifier to signify when the state is time invariant.

\medskip

Before we proceed to discussing the various assumptions of the model, it is worth making a few comments about our definition of rationalizability. First, observe that this criterion requires choices to be strictly optimal (captured by the strict inequality in point i). This assumption is made to rule out indifference that would make the analysis trivial. If the strict inequality in point i were to be replaced by a weak inequality, then any choice data could be rationalized by a constant belief and common utilities $p_t=\overline{u}_1=\cdots=\overline{u}_I$ for all $t\in \mT$. Second, implicit in point ii is the assumption that the realized signals $s_t$ lie in the support of the experiment so that Bayes' rule can be used. In subsequent sections, we will discuss how non-Bayesian belief updating can be accommodated.

\subsection{Discussion of the model}\label{sec:model_discussion}

The payoffs in the baseline model capture situations where decision-makers are uncertain about which course of action is best and differ in their disutility for type I and II errors. Such preferences are common in economics and political science. Indeed, the model as a whole is fairly canonical in the literature on committee voting. Here, the state captures the correct decision (for instance, maintain the status quo or introduce a reform, convict or acquit a defendant in a court) and the actions are the votes. The repeated choices could either reflect votes on the same issue (deciding whether or not to implement a project that was postponed in a previous vote) or different observationally identical scenarios (such as judges deciding distinct cases with the same basic facts). We now discuss our modeling choices, starting with the assumption of common information.

\medskip 

\noindent{\sc Common signals:} As we discussed in the introduction, there are several multiple-agent decision problems where it is natural to assume that agents learn via common signals. In other contexts such as committee voting, the literature considers both models like ours with common information and settings where agents receive private signals prior to making choices. Common information is frequently assumed (such as in \citealp{li2001,alonso2016}) when the aim is to analyze the effect of information on committee outcomes. Surveys that describe the foundational work in this strand of the literature are \citet{gerling2005} and \citet{li2009}. Note that, even if the agents receive private signals, they often have the opportunity to communicate them to others or, in other words, to deliberate prior to making choices.

Moreover, there is a simple theoretical reason to focus attention on common signals. Without additional constraints, \emph{all} choice data  are rationalizable if the model were to allow for private signals (even with a common utility function). In a nutshell, private signals make it possible to decouple the choice problems agent-by-agent. Without additional constraints, we can then generate arbitrary sequences of individual beliefs  and, thus, rationalize all choice data. We refer to \cref{sec:comonotone} for a model with non-trivial testable implications, where agents receive private signals. 

\medskip

\noindent {\sc Realizations versus choice probabilities:} As we stressed in the introduction, a strength of our framework is that choice data consists of actual choices as opposed to choice probabilities. This implies that the data requirements to test our model are relatively modest. In particular, this means that we neither need to observe the agents choosing repeatedly in the identical information environment nor do we need to aggregate data across different agents. 

The strengths and weaknesses of our framework can be clearly articulated by contrasting with a related decision theory paper. \citet{lu2016} considers an agent who chooses an alternative after receiving some private information (unobserved by the analyst) and provides a theory in which the information can be fully identified from the random choice rule. For identification, he requires the analyst to observe the choice probabilities from \textit{every} menu of possible acts (mappings from states to lotteries over prizes). By contrast, we only require the analyst to observe the \textit{realized} choices from a \textit{single menu} of two acts; alternative $x$ (and similarly $y$) can be viewed as an act that generates payoff $u(x,\omega)$ in each state $\omega$. Of course, the downside of an analysis on such parsimonious data is that, unlike \cite{lu2016}, we cannot identify the information.

\medskip

\noindent {\sc Separate decision problems vs. equilibrium:} We presented our choice data as arising from the separate decisions of $I$ agents. This is the appropriate interpretation for several of the examples from the introduction including that of the two employers. In the context of committee decision making, this corresponds to sincere voting where members vote for their preferred option (and ignore strategic considerations). However, in the case of two alternatives and public information (as in our baseline model), sincere voting is also an equilibrium of the one-shot voting game where alternative $x$ is chosen only when more than a given exogenous threshold of voters choose $x$.

In the repeated interaction that is implicit in our choice data, sincere voting will also be an equilibrium of the repeated game since it is a stage game equilibrium (although, of course, there may be multiple equilibria). That said, our model also captures certain, more complex dynamic decision making as a special case. We discuss this next.

\medskip

\noindent {\sc Dynamic optimization:} The agents we defined above myopically maximized their per period utility. This assumption is commonly made in the revealed preference literature that builds on \citet{afriat1967} where the aim is to derive the testable implications of various models of consumption on price quantity data (see \citealp{chambers2016} for a description of the literature). In our setting, behaving myopically is also dynamically optimal if the agents' actions do not affect the evolution of the state and the information they receive. This is true for certain applications like judicial decision making (which we discuss in more detail in \cref{sec:app_committees}) where the decision to convict or acquit does not shed any additional light on the guilt of defendant.

Importantly, certain classes of dynamic models in which actions affect the arrival of information also lead to behavior that is observationally indistinguishable. Loosely speaking, the solution to a Bellman equation can result in cutoff behavior in which agent $i$'s action choice depends on whether the public belief lies above or below a cutoff $\overline{u}_i$. In other words, even though agents are actually dynamically optimizing, their behavior is covered by our rationalization criterion (\cref{def:rat}). Recent examples of such dynamic political economy models that fit our framework are \citet{gul2012} and \citet{chan2018}.

\medskip

\noindent {\sc Richer choice data:} The baseline model assumes that the researcher only observes the agent's choices. We can also consider richer choice data $\mX = (x_{i,t},\theta_t)_{(i,t) \in \mI \times \mT}$, where the researcher also observes a vector $\theta_t \in \Theta_t$ of characteristics about the choice problems at each period $t$. We can generalize our definition of rationalization to experiments $\pi_t : \{\omega^x,\omega^y\} \rightarrow \Delta(S_t \times \Theta_t)$ and transitions $\gamma_t:  \{\omega^x,\omega^y\} \times \Theta_{t-1} \rightarrow \Delta( \{\omega^x,\omega^y\})$, without affecting our main results. If the richer choice data $\mX = (x_{i,t},\theta_t)_{(i,t) \in \mI \times \mT}$ is rationalizable, the simpler choice data $ (x_{i,t})_{(i,t) \in \mI \times \mT}$ is rationalizable with $\widehat{\pi}_{t}(s_t|\omega_t):= \pi_{t}(s_t,\theta_t|\omega_t)/\pi_t(\theta_t|\omega_t)$ and $\widehat{\gamma}_t(\omega_t|\omega_{t-1}):=  \gamma_t(\omega_t|\omega_{t-1},\theta_{t-1})$, and we can similarly do the converse. In \cref{sec:app_discrimination}, we demonstrate how richer choice data can be used to detect evidence of discrimination; recall that the example at the end of the introduction had this feature as the analyst observed the gender of the applicants.

\medskip 

To summarize, we view our model to be a simple reduced form representation of dynamic behavior that is flexible enough to encompass several distinct dynamic settings as special cases.

\section{The testable implications of the baseline model} \label{sec:results_baseline}

We first present a simple example, which shows that not all choice data are rationalized. 

\begin{example}\label{eg:main_example}
	Consider the choice data $\mX$ of two individuals $i$, $j$ and two time periods $t_1$, $t_2$ represented in the following table:
		\begin{center}
		\begin{tabular}{|c|c|c|}\hline
			& $i$ & $j$     \\ \hline
			$t_1$ & $x$ & $y$     \\ \hline
			$t_2$ & $y$ & $x$\\ \hline
		\end{tabular}
	\end{center}

	\medskip
	
	\noindent In the table, choices in the rows correspond to a given time period and those in the the columns correspond to the agents.
	
	We argue that this choice data cannot be rationalized. So suppose, for the purpose of contradiction, that it can. This implies that we can find utilities and beliefs that satisfy
	\begin{align*}
		& p_{t_1} >\overline{u}_i > p_{t_2}, \\
		& p_{t_1} < \overline{u}_j < p_{t_2},
	\end{align*}
	which is clearly a contradiction. In words, agent $i$'s actions imply that the belief must be strictly lower in observation $t_2$ compared to $t_1$ whereas agent $j$'s actions imply the exact opposite.  
\end{example}

It is worth discussing a few features of the simple argument we employed in the above example. Note that it did not even mention the state transition probabilities. Moreover, this argument did not require the agents to be Bayesian.  As long as they share a common belief, they could be learning in \textit{any} way and the data in \cref{eg:main_example} would still not be rationalizable. Finally, the time periods $t_1$ and $t_2$ need not be consecutive or ordered chronologically.

\medskip

The following definition gives a name to the pattern of choices in \cref{eg:main_example}.

\begin{definition}\label{def:cycles}
	Choice data $\mX$ has a \textbf{cycle} if there are distinct time periods $(t_1,t_2)\in \mT \times \mT$ and agents $(i,j) \in \mI \times \mI$ such that $x_{i,t_1}=x_{j,t_2}=x$ and $x_{i,t_2}=x_{j,t_1}=y$.
\end{definition}

\cref{eg:main_example} shows that a necessary condition for rationalization is the absence of cycles. The following result shows that this condition is also sufficient. In other words, the choice data described in \cref{eg:main_example} is the \textit{only} pattern of choices that cannot be rationalized. 

\medskip

\begin{theorem}\label{thm:baseline}
	Choice data $\mathcal{X}$ is rationalizable if, and only if, it has no cycles.
\end{theorem}

\medskip

We present a simple proof of \cref{thm:baseline} as it will be useful for the discussion that follows. 

\begin{proof} $(\Rightarrow).$ Already shown. 

\noindent $(\Leftarrow).$ We define the following binary relation $\succcurlyeq$ on the set of agents $\mI$. For any $(i,j)\in\mI \times \mI$, 
	$$j \succcurlyeq i \qquad \text{ if, for all $t\in\mT$, we have } \qquad x_{j,t}=x\implies x_{i,t}=x.$$
	In words, $j\succcurlyeq i$ if whenever agent $j$ chooses $x$, so does agent $i$. 
	
	We  argue that, since $\mX$ has no cycles, $\succcurlyeq$ is complete. By contradiction, assume that there exists a pair $(i,j)$ such that $j\not\succcurlyeq i$ and $i\not\succcurlyeq j$. Since $j\not\succcurlyeq i$, there exists $t \in \mT$ such that $x_{j,t}=x$ and $x_{i,t}=y$. Since $i\not\succcurlyeq j$, there exists $t' \in \mT$ such that $x_{i,t'}=x$ and $x_{j,t'}=y$. We thus have a cycle, the required contradiction. \medskip
	
	We now argue that  $\succcurlyeq$ is transitive. This is because for any $(i,j,k) \in\mI$, if $a_{j,t}=x\implies a_{i,t}=x$ for all $t\in\mT$ (and so $j\succcurlyeq i$) and $a_{k,t}=x\implies a_{j,t}=x$ for all $t\in\mT$ (and so $k\succcurlyeq j$), then we must have $a_{k,t}=x\implies a_{i,t}=x$ for all $t\in\mT$ (and so $k\succcurlyeq i$).\medskip

	Since $\succcurlyeq$ is complete and transitive, we can find numbers $\overline{u}_i\in (0,1)$ for all $i\in\mI$ (by invoking a textbook utility representation theorem for finite alternatives) such that
	$$\overline{u}_j \geq \overline{u}_i\;\; \text{ if and only if\;} \;\;j\succcurlyeq i.$$\medskip 
	
	We now argue that if $x_{i,t}=x$ and $x_{j,t}=y$, then $\overline{u}_j > \overline{u}_i$. By contradiction, if $\overline{u}_i \geq \overline{u}_j$, then $i \succcurlyeq j$ and, therefore, $x_{i,t}=x$ implies that $x_{j,t}=x$, a contradiction. It follows that $$\min\left\{ \overline{u}_j \;| \; x_{j,t}=y\} > \max\{ \overline{u}_i \;| \; x_{i,t}=x\right\}$$ for all $t\in\mT$.\medskip
	
	Finally, let   $p_0=\frac12$ and, for each $t\in \mT$, choose any $p_t\in(0,1)$ such that 
	$$\min\left\{ \overline{u}_i \;| \; x_{i,t}=y\}>p_t>\max\{ \overline{u}_j \;| \; x_{j,t}=x\right\}.$$ By construction, the data is rationalized with these beliefs and thresholds.\medskip 
	
	We can easily construct utility functions with the required thresholds, for instance, by setting $u_i(x,\omega^x)=u_i(y,\omega^y)=1$, $u_i(x,\omega^y)=1/2$ and $u_i(y,\omega^x)= (3 \overline{u}_i-1)/2\overline{u}_i$.\medskip
	
	It remains to construct the transitions probabilities, signals and experiments that generate the given beliefs. So first, we pick transition probabilities $\gamma_t(\omega^x|\omega^x)=\gamma_t(\omega^y|\omega^y)=1$ for all $t\in\mT$ that correspond to the state being time invariant. Then, for all $t$, pick an experiment with binary signals $S_t=\{s,s'\}$ and choose any $\pi_t(s|\omega^x) \in (0,1)$ and $\pi_t(s|\omega^y)\in (0,1)$ such that
	$$\frac{\pi_t(s|\omega^x)}{\pi_t(s|\omega^y)}=\frac{p_{t}}{1-p_{t}}\frac{1-p_{t-1}}{p_{t-1}}$$
	and assign $\pi_t(s'|\omega^x)=1-\pi_t(s|\omega^x)$ and $\pi_t(s'|\omega^y)=1-\pi_t(s|\omega^y)$.\medskip
	
	It is routine to verify that $p_t$ is obtained from $p_{t-1}$ by Bayesian updating when signal $s$ is realized. This completes the proof.
\end{proof}

\medskip

We end this section by discussing some implications of \cref{thm:baseline}. First, note that the proof shows that, for rationalizable choice data, the actions reveal information about the preference cutoffs via the relation $\succcurlyeq$. Thus, we can determine when an agent $i$ has a higher preference cutoff than $j$ (if $i \succcurlyeq j$ and $j\not\succcurlyeq i$). This preference information is potentially useful; for instance in \cref{sec:app_discrimination}, we argue that it can be used to uncover evidence of discrimination.\medskip

Second,  observe that the proof shows that every rationalizable data set is also rationalizable with a time-invariant state. In other words, while the model is refutable, we cannot distinguish between environments where the state is time varying or invariant. Similarly, the proof shows that any beliefs that rationalize the choice data can be obtained by Bayesian updating. Therefore, Bayesian updating itself is not  testable. So, as long as agents have common beliefs, they could be using \textit{any} method to form their beliefs. All we need is for the agents to change their beliefs as they receive informative signals.\medskip 

While these observations follow almost immediately from our definition of rationalizability, we nonetheless view them to be important conceptual insights of our framework. They demonstrate the limits of the behavioral analysis of learning that can be conducted from finite choice data which is, more often than not, the available field data. In this sense, these insights are akin to Afriat's theorem \citep{afriat1967} that shows any price quantity data that is consistent with utility maximization with a locally nonsatiated utility function is also consistent with utility maximization with a concave utility (so the convexity of preferences is not testable).

\section{Generalizing the baseline model}\label{sec:generalizations}

This section generalizes the baseline model along several different directions and derives the analogue of \cref{thm:baseline} for each of these generalizations. Each subsection generalizes the baseline model along one particular dimension, but it will be clear that we can also simultaneously allow for all, or a subset of, these distinct generalizations.

\subsection{General preferences}\label{sec:gen_prefs}

When we defined rationalization in \cref{sec:model}, we assumed that utilities were such that agents prefer to match their action to the underlying state. Needless to say, this assumption might not always be appropriate. Note that allowing for arbitrary utilities does not change the cutoff structure of optimal behavior: an agent chooses $x$ if his belief is above or below a cutoff. Hence, with general preferences, we can define rationalization as follows.

\begin{definition}
	The choice data $\mathcal{X}$ is \textbf{rationalizable with general preferences} if, for all $i\in\mI$ and $t\in\mT$, there exist utility functions $u_i$, constants $\eta_i\in\{-1,1\}$, transition probabilities $\gamma_t$, prior belief $p_{0}$, beliefs $p_{t}$, sets of signals $S_t$, experiments $\pi_t: \{\omega^x,\omega^y\} \rightarrow \Delta(S_t)$ and realized signals $s_t \in S_t$ such that 
	\begin{enumerate}[label=\roman*.]		
		\item \textit{Actions maximize utility:} $\eta_i p_{t} > \overline{u}_i$ when $x_{i,t}=x$ (resp., $\eta_i p_{t} < \overline{u}_i$ when $x_{i,t}=y$).
		\item \textit{Bayesian Updating:} $p_{t}$ is derived from $p_{t-1}$ by Bayes' rule.
	\end{enumerate}
\end{definition}
The difference is captured in point i. The optimal behavior of an expected utility maximizer with general preferences is characterized by two numbers $\overline{u}_i$ and $\eta_i$; the former determines the cutoff belief and the latter whether agents prefer to match or mismatch the state.

The characterization of choice data that is rationalizable with general preferences follows immediately from \cref{thm:baseline}. To formally state the result, we need one additional piece of notation. Let $\kappa_i:\{x,y\}\to \{x,y\}$ be a permutation (bijection) on the actions of agent $i$, and $\kappa=(\kappa_i)_{i\in \mI}$ be a profile of permutations, which we refer to simply as a \textit{permutation}. For any permutation $\kappa$, let $\mX^{\kappa}:=(x^{\kappa}_{i,t})_{(i,t) \in \mI \times \mT}$ be the \textit{permuted choice data} obtained from $\mathcal{X}$ by permuting the choices; that is, $x^{\kappa}_{i,t}=\kappa_i(x_{i,t})$ for all $(i,t)$.

\begin{corollary}\label{cor:gen_prefs}
	Choice data $\mathcal{X}$ is rationalizable with general preferences if, and only if, we can find a permutation $\kappa$ such that the permuted choice data $\mathcal{X}^{\kappa}$ has no cycles.
\end{corollary}

This result is very intuitive. If the data is rationalizable with general preferences, it must mean that we can switch the choices of every agent for whom $\eta_i=-1$ and the permuted choice data must be rationalizable in the sense of \cref{thm:baseline}. The next example generalizes \cref{eg:main_example} to show that not every choice data set is rationalizable with general preferences.

\begin{example}\label{eg:gen_prefs}
	Consider the choice data $\mX$ of two individuals $i$, $j$ and four time periods represented in the following table:
	\begin{center}
		\begin{tabular}{|c|c|c|}\hline
			& $i$ & $j$     \\ \hline
			$t_1$ &\cellcolor{blue!25} $x$ & \cellcolor{blue!25} $y$     \\ \hline
			$t_2$ &\cellcolor{blue!25} $y$ & \cellcolor{blue!25} $x$\\ \hline
			$t_3$ & $x$ & $x$\\ \hline
			$t_4$ & $y$ & $y$\\ \hline
		\end{tabular}
	\end{center}		

	\medskip
	
	Observe that the choices in the shaded cells are identical to those of \cref{eg:main_example} and constitute a cycle. \cref{cor:gen_prefs} states that the data cannot be rationalized with general preferences when there is a cycle for every permutation of the choice data. Clearly, the cycle in periods $t_1$ and $t_2$ can be removed by permuting the choices of either agent. However, this would then introduce a cycle into periods $t_3$ and $t_4$. Below, we display the choice data from the three possible profiles $(\kappa_i,\kappa_j)$ of permutations other than the identity permutation (which does not change the choice data): from left to right, the permutations are on the choices of agent $i$ only, agent $j$ only and both. Shaded cells highlight the cycles and this demonstrates that the choice data in this example is not rationalizable with general preferences.
	
	\medskip
	
	\begin{center}		
		\begin{tabular}{|c|c|c|}\hline
			& $i$ & $j$     \\ \hline			
			$t_1$ & $y$ & $y$\\ \hline
			$t_2$ & $x$ & $x$\\ \hline
			$t_3$ &\cellcolor{blue!25} $y$ & \cellcolor{blue!25} $x$     \\ \hline
			$t_4$ &\cellcolor{blue!25} $x$ & \cellcolor{blue!25} $y$\\ \hline
		\end{tabular}	
		\qquad\qquad\qquad
		\begin{tabular}{|c|c|c|}\hline
			& $i$ & $j$     \\ \hline
			$t_1$ & $x$ & $x$\\ \hline
			$t_2$ & $y$ & $y$\\ \hline
			$t_3$ &\cellcolor{blue!25} $x$ & \cellcolor{blue!25} $y$ \\ \hline
			$t_4$ &\cellcolor{blue!25} $y$ & \cellcolor{blue!25} $x$\\ \hline
		\end{tabular}
		\qquad\qquad\qquad
		\begin{tabular}{|c|c|c|}\hline
			& $i$ & $j$     \\ \hline
			$t_1$ &\cellcolor{blue!25} $y$ & \cellcolor{blue!25} $x$   \\ \hline
			$t_2$ &\cellcolor{blue!25} $x$ & \cellcolor{blue!25} $y$\\ \hline
			$t_2$ & $y$ & $y$\\ \hline
			$t_2$ & $x$ & $x$\\ \hline
		\end{tabular}
	\end{center}
\end{example}

\subsection{More than two alternatives}\label{sec:mult_actions}

There are many situations where there may be more than two alternatives despite the state being binary. For instance, a member of a committee may choose to abstain as opposed to voting for/against a proposal. 

\medskip

We begin by discussing the case where the alternatives are ordered. So suppose there are $N$ alternatives $\{y_1,\dots, y_N\}$ and each agent $i$ picks $x_{i,t}\in \{y_1,\dots, y_N\}$ in period $t$. Every agent $i$'s preferences are captured by cutoffs $0=:\overline{u}_{i,1}<\cdots<\overline{u}_{i,N+1}:=1$ and the agent chooses action $x_{i,t}=y_n$ whenever $\overline{u}_{i,n}<p_t<\overline{u}_{i,n+1}$.\medskip

Because agents are expected utility maximizers, for any utility function $u_i$, their optimal choices will always take the above cutoff form for \textit{some} ordering of the alternatives. Therefore, we are implicitly making two assumptions above. The first is that all agents order alternatives in the same way (although their belief cutoffs might differ). The second is that alternative $y_1$ is picked at the lowest beliefs and vice versa for $y_N$. As with \cref{thm:baseline}, these assumptions can be dispensed with and we describe how to do so immediately after the statement of \cref{thm:mult_actions} below.

\medskip

Rationalization with multiple alternatives is analogously defined to \cref{def:rat}.
\begin{definition}
	The choice data $\mathcal{X}$ is \textbf{rationalizable with multiple alternatives} if, for all $i\in\mI$ and $t\in\mT$, there exist utility functions $u_i$, transition probabilities $\gamma_t$, prior belief $p_{0}$, beliefs $p_{t}$, sets of signals $S_t$, experiments $\pi_t: \{\omega^x,\omega^y\} \rightarrow \Delta(S_t)$ and realized signals $s_t \in S_t$ such that 
	\begin{enumerate}[label=\roman*.]		
		\item \textit{Actions maximize utility:} $\overline{u}_{i,n}<p_t<\overline{u}_{i,n+1}$, when $x_{i,t}=y_n$,
		\item \textit{Bayesian Updating:} $p_{t}$ is derived from $p_{t-1}$ by Bayes' rule.
	\end{enumerate}
\end{definition}

Similarly, we can define what it means to have a cycle when there are multiple alternatives. We use the same terminology as this is the obvious generalization of \cref{def:cycles}.
\begin{definition}\label{def:cycles_mult_actions}
	Choice data $\mX$ with multiple alternatives has a \textbf{cycle} if there are distinct time periods $(t_1,t_2)\in \mT \times \mT$ and agents $(i,j)\in \mI \times \mI$ such that $x_{i,t_1}=y_{n_{i,t_1}}$, $x_{j,t_1}=y_{n_{j,t_1}}$ and $x_{i,t_2}=y_{n_{i,t_2}}$, $x_{j,t_2}=y_{n_{j,t_2}}$ where $n_{i,t_1}>n_{i,t_2}$ and $n_{j,t_1}<n_{j,t_2}$.
\end{definition}
\noindent Note that the above definition implicitly leverages the order of alternatives that we fixed when we defined preferences.

\medskip

The absence of cycles is a necessary condition for rationalization even with multiple alternatives. To see this, observe that when agent $i$ switches from action $y_{n_{i,t_1}}$ to $y_{n_{i,t_2}}$ where $n_{i,t_1}>n_{i,t_2}$, this implies that we must have $p_{t_1}>\overline{u}_{i,n_{i,t_1}}>p_{t_2}$. However, the choices of agent $j$ imply exactly the opposite inequality for the beliefs $p_{t_1}<\overline{u}_{j,n_{j,t_2}}<p_{t_2}$, which generates the required contradiction. Indeed, like \cref{thm:baseline}, this condition is also sufficient.

\begin{theorem}\label{thm:mult_actions}
	Choice data $\mathcal{X}$ is rationalizable with multiple alternatives if, and only if, it has no cycles.
\end{theorem}

\noindent The proof of \cref{thm:mult_actions} can be found in \cref{app:mult_actions}.

\medskip

We can relax the particular ordering of the actions that we chose. As we discussed above, any utility function $u_i$ implies cutoff behavior for some order of actions and so, as in \cref{cor:gen_prefs}, we can define agent specific permutations $\kappa_i:\{y_1,\dots,y_n\}\to \{y_1,\dots,y_n\}$ and consider the permuted choice data $\mathcal{X}^{\kappa}$. Then, the choice data is rationalizable with multiple alternatives and general preferences if, and only if, we can find a permutation $\kappa$ such that the permuted choice data $\mathcal{X}^{\kappa}$ has no cycles.

\subsection{Relaxing common beliefs}\label{sec:comonotone}

Perhaps the strongest assumption in the baseline model is that agents have common beliefs.  This section relaxes this assumption.  As we mentioned earlier, every choice data set is rationalizable if we allow for arbitrary private signals and, thus, for arbitrary private beliefs.  There needs to be some relation between the information of the agents in order for the model to have non-trivial testable implications.\medskip

The restriction we impose is that agents receive  ``co-monotone'' private signals. Informally, an experiment is co-monotone if the joint distribution of the signals is such that agents agree on the state that their individual signal is more likely to arise from, but disagree on the strength of the signal. As as example, this assumption captures a deliberative committee where agents all agree that the information is good news for one of the two states, but disagree in terms of how good the news actually is.

\medskip

Formally, the experiment $\pi_t: \{\omega^x,\omega^y\} \rightarrow \Delta(\times_{i \in \mI} S_{i,t})$ is \textit{co-monotone} if for all  profiles of signals $(s_{1,t},\dots,s_{I,t})$ in the support of $\pi_t$, 
\[[\pi_{i,t}(s_{i,t}|\omega^x)-\pi_{i,t}(s_{i,t}|\omega^y)][\pi_{j,t}(s_{j,t}|\omega^x)-\pi_{j,t}(s_{j,t}|\omega^y)]\geq 0 \]
for all $(i,j)$, where $\pi_{i,t}(\cdot|\omega)$ is the marginal of $\pi_t(\cdot|\omega)$ over $S_{i,t}$.\footnote{The support of the kernel $\pi_t$ is the union of the supports of each probability $\pi_t(\cdot|\omega)$.}  Experiments are \textit{strictly} co-monotone if the above inequalities are all strict.

\medskip

Before defining the rationalization criterion, we present the main implication of assuming co-monotone experiments. To do so, we need some additional notation. We let $p_{i,t}$ be the belief about the event $[\omega_t=\omega^x]$ of agent $i$ at period $t$ after having received the private signal $s_{i,t}$. Similarly, we let 
$$q_{i,t+1}:=p_{i,t}\gamma_{t+1}(\omega^x|\omega^x)+(1-p_{i,t})\gamma_{t+1}(\omega^x|\omega^y)$$
be the belief about the event $[\omega_{t+1}=\omega^x]$ of agent $i$ at period $t+1$ prior to receiving  the private signal $s_{i,t+1}$.
An immediate implication of assuming co-monotone experiments is that, for any two Bayesian agents $i$ and $j$, 
\begin{equation}\label{eq:comon}
	[p_{i,t}> q_{i,t} \;\; \implies \;\; p_{j,t}\geq  q_{j,t}] \;\text{and\;} [p_{i,t}< q_{i,t} \;\; \implies \;\; p_{j,t}\leq  q_{j,t}]. 
\end{equation}
 In words, the private beliefs are co-monotone. 
Conversely, for all posteriors $p_{i,t}$ and $p_{j,t}$ that satisfy \eqref{eq:comon}, we can find a co-monotone experiment and signals $(s_{i,t},s_{j,t})$, which generate these posteriors. (See the proof of  \cref{thm:comon_inv} for details.)  
\medskip

We now define the rationalization criterion. Throughout, we write $S_t$ for the set of signals $\times_{i \in \mI} S_{i,t}$.

\begin{definition}\label{def:comonotone}
	The choice data $\mathcal{X}$ is \textbf{rationalizable with (strictly) co-monotone experiments} if there exist utility functions $u_i$, transition probabilities $\gamma_t$, prior $p_{0}$, beliefs $p_{i,t}$, sets of signals $S_{i,t}$, (strictly) co-monotone experiments $\pi_{t}: \{\omega^x,\omega^y\} \rightarrow \Delta(S_t)$ and realized signals $s_t \in S_t$ such that 
	\begin{enumerate}[label=\roman*.]
		\item \textit{Actions maximize utility:} $p_{i,t} >  \overline{u}_i$ if $x_{i,t}=x$ (resp., $p_{i,t} <  \overline{u}_i$ if $x_{i,t}=y$).
		\item \textit{Bayesian Updating:} $p_{i,t}$ is derived from $p_{i,t-1}$ by Bayes' rule conditioning on $s_{i,t}$.
	\end{enumerate}
\end{definition}

We stress that all the information an agent receives is encoded in the private signal. Since private signals are not correlated with past decisions, states and signals, this implies that agents neither observe the past decisions of others nor their private signals. This assumption is appropriate in some applications, less so in others. We can allow for general experiments, which condition on the current state \emph{and} past signals, actions, and states. As long as we restrict attention to co-monotone experiments, the same testable implications would obtain.\medskip 

 We first argue that the model has testable implications when the state is \textit{time-invariant}. To do so, we revisit \cref{eg:main_example} making one minor alteration.
\begin{example}\label{eg:comonotone}
	Consider the choice data $\mX$ of two individuals $i$ and $j$ and two \textit{consecutive} time periods $t$ and $t+1$ represented in the following table:
	\begin{center}
		\begin{tabular}{|c|c|c|}\hline
			& $i$ & $j$     \\ \hline
			$t$ & $x$ & $y$     \\ \hline
			$t+1$ & $y$ & $x$\\ \hline
		\end{tabular}
	\end{center}

	\medskip
	
	\noindent Recall that \cref{eg:main_example} had exactly this pattern of choices, but with arbitrary time periods.
		
	We argue that this choice data cannot be rationalized with co-monotone experiments and a time-invariant state. Suppose to the contrary that we could. This implies that we can find thresholds and beliefs for agent $i$ that satisfy
	\begin{align*}
		& p_{i,t} >\overline{u}_i > p_{i,t+1}.
	\end{align*}
	Since agent $i$'s belief decreases, it must be the case that the information she receives satisfies $\pi_{i,t}(s_{i,t}|\omega^x)<\pi_{i,t}(s_{i,t}|\omega^y)$. Since experiments are co-monotone, agent $j$'s information must satisfy $\pi_{j,t}(s_{j,t}|\omega^x)\leq \pi_{j,t}(s_{j,t}|\omega^y)$ for all signals $s_{j,t}$ that agent $j$ can receive with positive probability. This implies that agent $j$'s belief cannot increase and so we cannot have
	\begin{align*}
		& p_{j,t} < \overline{u}_j < p_{j,t+1},
	\end{align*}
	the required contradiction.
\end{example}

\medskip

The following definition provides a name to such consecutive time period cycles.
\begin{definition}\label{def:con_cycles}
	Choice data $\mX$ has a \textbf{consecutive cycle} if there are agents $(i,j) \in \mI \times \mI$ and a time period $t\in \mT$ such that $x_{i,t}=x$, $x_{j,t}=y$ and $x_{i,t+1}=y$, $x_{j,t+1}=x$.
\end{definition}

\medskip

\cref{eg:comonotone} shows that a necessary condition for rationalization is the absence of consecutive cycles in the choice data. The next result shows that this condition is also sufficient. The theorem actually shows a stronger result. First, any data that is rationalizable with co-monotone experiments is also rationalizable with strictly co-monotone experiments. Second, if the choice data is rationalized, then it is rationalizable with an arbitrary profile of utility functions. Put differently, this implies that, if we relax the assumption of common information, then choices no longer provide any ranking of the agents' cutoff beliefs. 

\begin{theorem}\label{thm:comon_inv}
	Fix any arbitrary profile of utility functions  $(u_i)_{i \in \mI}$. Given the profile $(u_i)_{i \in \mI}$, choice data $\mathcal{X}$ is rationalizable with (strictly) co-monotone experiments and a time-invariant state if, and only if, it has no consecutive cycles.
\end{theorem}
\begin{proof}
	Since we have already argued the necessity, we only need to argue the sufficiency.  Fix any arbitrary profile of utility functions  $(u_i)_{i \in \mI}$. The proof is by induction on $T$.  
	
	If $T=1$,  there are no consecutive cycles. We therefore need to argue that we can rationalize all choice data. Given the thresholds $(\overline{u}_i)_{i \in \mI}$, choose $p_0$ and $p_{i,1}$ as follows. For all $i$ such that $x_{i,1}=x$, choose $p_{i,1} > \overline{u}_i$. Similarly, for all $i$ such that $x_{i,t}=y$, choose  $\overline{u}_i >p_{i,1}$. Choose $p_0 < \min_{i \in \mI} p_{i,1}$. For all $i$, choose positive numbers $(\pi^x_{i,1},\pi_{i,1}^y )$ such that $\sum_{i} \pi_{i,1}^x <1$, $\sum_{i} \pi_{i,1}^y<1$, and
	\[\frac{\pi_{i,1}^x}{\pi_{i,1}^y}= \frac{p_{i,1}}{1-p_{i,1}}\frac{1-p_0}{p_0} >1.\]
	Such numbers exist. Moreover, since $\pi_{i,1}^x/\pi_{i,1}^y > 1$ for all $i$, we have that  
	\[[\pi_{i,1}^x-\pi_{i,1}^y][\pi_{j,1}^x-\pi_{j,1}^y]> 0,\]
	for all $(i,j)$. From \cref{app:lemma_expe} in  \cref{app:private_signals}, there exist a strictly co-monotone experiment and a profile of signals $(s_{i,1},\dots,s_{I,1})$ such that $(\pi_{i,1}^x, \pi_{i,1}^y)= (\pi_{i,1}(s_{i,1}|\omega^x), \pi_{i,1}(s_{i,1}|\omega^y))$.  Hence, upon receiving the signal $s_{i,1}$, agent $i$'s belief is $p_{i,1}$. \medskip 
	
	As the induction hypothesis, suppose we can rationalize any $T$ period choice data that has no consecutive cycles. For the induction step, we will argue that we can also do this for any $T+1$ period choice data that has no consecutive cycles.\medskip 
	
	So consider choice data $\mX:=(x_{i,t})_{(i,t) \in \mI \times \mT}$ with no consecutive cycles and $|\mT|=T+1$. By the induction hypothesis, we can rationalize the first $T$ observations $(x_{i,t})_{(i,t) \in \mI \times \{1,\dots,T\}}$. Let the period $T$ beliefs be $p_{i,T}$.\medskip
	
	If the choices of all agents are the same in periods $T$ and $T+1$, that is, $x_{i,T}=x_{i,T+1}$ for all $i\in\mI$, then the choice data is rationalized by all agents receiving uninformative signals.\medskip
		
	It remains to consider the case where at least one agent changes their choice. So, assume that there is an agent $\tilde{i}$ who switches from $x_{\tilde{i},T}=x$ to $x_{\tilde{i},T+1}=y$. (The case where an agent $\tilde{i}$ switches from $a_{\tilde{i},T}=y$ to $a_{\tilde{i},T+1}=x$ is treated analogously and, therefore, omitted.) Since there are no consecutive cycles,  there is no agent $j$ switching from $x_{j,T}=y$ to $x_{j,T+1}=x$. So for all $i\in \mI$, choose any $p_{i,T+1}$ that satisfy
	$$\begin{array}{ll}
		p_{i,T}>p_{i,T+1}>\overline{u}_i \;\; & \text{ if } x_{i,T}=x_{i,T+1}=x, \\
		\overline{u}_i > p_{i,T}>p_{i,T+1} \;\; & \text{ if } x_{i,T}=x_{i,T+1}=y, \\
		p_{i,T}>\overline{u}_i > p_{i,T+1} \;\; & \text{ if } x_{i,T}=x,\;\;x_{i,T+1}=y.
	\end{array}$$
	Such period-$(T+1)$ beliefs can always be chosen because, by the induction hypothesis the period-$T$ beliefs satisfy $p_{i,T}> \overline{u}_i$ when $x_{i,T}=x$ (respectively, $p_{i,T}< \overline{u}_i$ when $x_{i,T}=y$).\medskip

	Observe that the constructed period $T+1$ beliefs are strictly lower than the period $T$ beliefs and hence, we can construct co-monotone experiments that generate these beliefs. For all $i$, choose positive numbers $(\pi_{i,T+1}^x,\pi_{i,T+1}^y)$ such that $\sum_i \pi_{i,T+1}^x<1$, 
	$\sum_i \pi_{i,T+1}^y<1 $, and
	$$\frac{\pi_{i,T+1}^x}{\pi_{i,T+1}^y}=\frac{p_{i,T+1}}{1-p_{i,T+1}}\frac{1-p_{i,T}}{p_{i,T}}<1.$$
	To complete the proof, we invoke \cref{app:lemma_expe} in  \cref{app:private_signals}.
	\end{proof}

\medskip

Unlike \cref{thm:baseline}, \cref{thm:comon_inv} requires the state to be time-invariant. Why should the state evolution matter in this case? In the case of common beliefs, by definition, beliefs must move in the same direction for all agents across observations. This is also true for co-monotone experiments as long as the state is time-invariant. When the state is time-varying however, this need not be true. To see this, suppose that the evolution of the state causes the belief $q_{i,T+1}$ to be lower than $p_{i,T}$ for all agents. Now, consider a co-monotone experiment where the belief of each agent increases, that is, $p_{i,T+1} > q_{i,T+1}$ for all $i$. The different strengths of the signals may, however, cause the belief $p_{i,T+1}$ of some agents to increase above their belief $p_{i,T}$, while for others it remains below. \cref{thm:comon_var} states that this is in fact enough to rationalizable all choice data. 

\begin{theorem}\label{thm:comon_var}
	Every choice data $\mathcal{X}$ is rationalizable with strictly co-monotone experiments (and a time-varying state).
\end{theorem}
\begin{proof}
The proof is by induction on $T$. Throughout, we fix arbitrary thresholds $(\overline{u}_i)_{i \in \mI}$. If $T=1$, all choice data can clearly be rationalized (as shown in the proof of \cref{thm:comon_inv}). \medskip 

As the induction hypothesis, suppose that we can rationalize all choice data with $T$ periods and, for the induction step, consider any choice data $(x_{i,t})_{(i,t) \in \mI \times \mT}$ with $|\mT|=T+1$ periods.  Since the choice data $(x_{i,t})_{(i,t) \in \mI \times \{1,\dots, T\}}$ is rationalizable, there are period-$T$ beliefs $p_{i,T}$ such that 
$p_{i,T} > \overline{u}_i$ when $x_{i,T}=x$, and $p_{i,T} < \overline{u}_i$ when $x_{i,T}=y$.\medskip

Choose any $0 < \tilde{q} <\min_{i\in\mI} \overline{u}_i$ and consider the following state transitions
$$\gamma_{T+1}(\omega^x|\omega^x)=1-\gamma_{T+1}(\omega^y|\omega^x)=\gamma_{T+1}(\omega^x|\omega^y)=1-\gamma_{T+1}(\omega^y|\omega^y)=\tilde{q}.$$
By construction of the transition probabilities, $q_{i,T+1}= \tilde{q}$ for all $i$, regardless of the beliefs $p_{i,T}$. The evolution of the state at period $T+1$ is independent of its realization period $T$. \medskip 

For all $i\in \mI$, choose any $p_{i,T+1}>\tilde{q}$ that satisfies
	$$\begin{array}{ll}
		p_{i,T+1}>\overline{u}_i \;\; & \text{ if } x_{i,T+1}=x, \\
		p_{i,T+1}<\overline{u}_i \;\; & \text{ if } x_{i,T+1}=y.
	\end{array}$$
These numbers exist because $\tilde{q}<\min_{i\in\mI} \overline{u}_i$.\medskip
	
Finally, as in the proof of \cref{thm:comon_inv}, we can construct a strictly co-monotone experiment $\pi_{T+1}:\{\omega^x,\omega^y\}\ \rightarrow \Delta(\times_i S_{i,t})$, which generates these beliefs.  This completes the proof. \end{proof}

The observant reader may have wondered why we assume that agents have the same prior belief and agree about state transitions. It should be clear from the proof of \cref{thm:comon_inv} that nothing changes if we allow the prior belief to vary by agent (since the induction step does not involve the prior belief). \cref{thm:comon_var} shows that, if we allow for time varying states, anything goes. Thus, any further relaxation that allows for heterogeneous prior beliefs or private state transitions (even those that are systematically related) would also result in there being no testable implications.

\medskip

Finally, we can generalize \cref{thm:comon_inv} to accommodate multiple alternatives and general preferences in exactly the same way we did for the baseline model. With multiple alternatives, a consecutive cycle would then just be a cycle in terms of \cref{def:cycles_mult_actions} over consecutive time periods. General preferences would once again be the equivalent of permuting the choice data. Obviously \cref{thm:comon_var} also generalizes and all choice data is rationalizable with strictly co-monotone experiments, a time-varying state and, either or both of, multiple alternatives and general preferences.

\medskip

We conclude this section with two final remarks. First, an alternative to assuming that agents receive private signals is to assume that all agents observe the same signal, but have different (co-monotone) beliefs about its likelihood. That is, agents observe the signal $s_t$ at period $t$, but disagree about the likelihood of $s_t$. Agent $i$ believes that the likelihood is $\pi_{i,t}(s_t|\omega)$ when the state is $\omega$. In turn, this translates into agent $i$ having the belief $p_{i,t}$. Differences in beliefs are thus coming from differences in personal views about the data generating process. This alternative interpretation does not alter our results. 
\medskip 

Second, we assume that the agents are Bayesian. Again, this is not needed. All we need is that the agents revise their beliefs in a way that satisfies \eqref{eq:comon}. There are several models in the behavioral economics literature that have this feature. We now describe a couple of them. \medskip

To start with, consider the following updating behavior first motivated by \cite{grether1980}. Posteriors are derived by the formula
\begin{align*}
	p_{i,t}&=\frac{	q_{i,t}\left(\pi_{i,t}(s_{i,t}|\omega^x)\right)^{\beta_{i,t}}}{q_{i,t}\left(\pi_{i,t}(s_{i,t}|\omega^x)\right)^{\beta_{i,t}}+\left(1-q_{i,t}\right)\left(\pi_{i,t}(s_{i,t}|\omega^y)\right)^{\beta_{i,t}}}
\end{align*}
where $\beta_{i,t}\geq 0$. Here, the interpretation is that agents distort (that is, either over or under react to) the signal via their individual, time specific parameter $\beta_{i,t}$. Observe that when $\beta_{i,t}=1$, this corresponds to standard Bayesian updating but the flexibility accommodates two common biases in belief updating: ``under-inference'' when $\beta_{i,t}<1$ and ``over-inference'' when $\beta_{i,t}>1$. 

Observe that $p_{i,t} > q_{i,t}$ implies that $\pi_{i,t}(s_{i,t}|\omega^x)> \pi_{i,t}(s_{i,t}|\omega^y)$ (resp., $p_{i,t} < q_{i,t}$ implies that $\pi_{i,t}(s_{i,t}|\omega^x)< \pi_{i,t}(s_{i,t}|\omega^y))$. Therefore, if agent $j$ follows the same updating rule (but with parameter $\beta_{j,t}$), his belief $p_{j,t}$ moves co-monotonically with agent $i$'s belief if the experiment is co-monotone.  Our characterization thus continues to hold. The significant theoretical and experimental literature that builds on this framework is described in several surveys; two recent examples are \citet{benjamin2019} and \citet{gabaix2019}.

\medskip

Another framework is that of ``coarse'' updating. Here, in every period $t$, each agent $i$ has belief cutoffs $0=:\beta^0_{i,t}<\beta^1_{i,t}<\cdots< \beta^{N_{i,t}}_{i,t}<\beta^{N_{i,t}+1}_{i,t}:=1$ and $N_{i,t}$ ordered beliefs $0\leq \hat{q}^1_{i,t} < \cdots<\hat{q}^{N_{i,t}}_{i,t}\leq 1$. Roughly speaking, belief updating works as follows. The updated belief is $p_{i,t}=\hat{q}^{n_i}_{i,t}$ if the belief obtained by Bayesian updating (given by the right side of equation \ref{eq:Bayes}) lies between the cutoffs $\beta^{n_i-1}_{i,t}$ and $\beta^{n_i}_{i,t}$. In words, agents update coarsely by assigning intervals of beliefs obtained by Bayesian updating to single beliefs. There are several models that have this feature. Two recent examples are \citet{wilson2014} and \citet{jakobsen2021}; the latter paper contains a detailed description of the literature. Note that a co-monotone experiment combined with such updating generates posteriors that are ordered in the sense of \eqref{eq:comon} and so this behavior will also be captured under \cref{def:comonotone}. Additionally, observe that we can allow all agents to use different updating rules with distinct behavioral biases as long as their beliefs satisfy the property \eqref{eq:comon}.

\section{Applications}\label{sec:applications}

\subsection{Testing for discrimination}\label{sec:app_discrimination}

Disentangling channels of discrimination is known to be empirically difficult (for instance, see the discussion in the recent survey of \citealp{guryan2013}). Our results can provide distinct suggestive evidence for different types of discrimination. To see this, consider richer choice data $(x_{i,t},\theta_{t})_{(i,t)\in\mI\times\mT}$, where $x_{i,t}$ is agent $i$'s decision to either approve or reject application $t$, and $\theta_t$ is a vector of observable characteristics (which are observable to both the agents and the analyst). As an example, prospective Ph.D. students routinely apply to several graduate schools, who either accept or reject their applications. In this example, the dataset consists of the approval decisions along with some observable characteristics of the applicants; for example, sex, age, or ethnicity. Other examples include applications for jobs, mortgages, and credit cards.\medskip 

For simplicity, suppose that there is a single characteristic $\boldsymbol{\theta}$, which only takes two values $\theta$ and $\theta'$ that, for example, correspond to being white and black or male and female. Here, the interpretation is that the applications are identical along all observable characteristic but $\boldsymbol{\theta}$.   Let $\theta$ be the characteristic corresponding to the ``favored'' group. \medskip

There are two main forms of discrimination: statistical and taste-based. We consider the former first. Consider the sub-sample where we condition on $[\theta_t=\theta]$, that is, we restrict attention to the observations having the favored characteristics. Assume that it can be rationalized. This means that we cannot reject the hypothesis that the agents share the same belief in this sub-sample. Now, consider the sub-sample where we condition on $[\theta_t=\theta']$. Suppose, we cannot rationalize this sub-sample. Then, this implies that, in contrast to the favored group, there are at least two agents $(i,j)$ who have different beliefs in the sub-sample of the ``disadvantaged'' group. This is suggestive evidence of \textit{statistical} discrimination.\medskip 

The evidence is, however, not fully conclusive. The two agents $(i,j)$ may have different beliefs because their information differs. We cannot reject that hypothesis with our simplest test. However, we can if we assume that the agents receive private but co-monotone signals. (We also need the state to be time-invariant, but this is a natural assumption when testing for discrimination where each observation corresponds to a different individual with the identical characteristic $\boldsymbol{\theta}$.)  Two additional remarks are worth making. First, we cannot definitely state whether it is agent $i$, agent $j$ or both, who are discriminating. We can only state that at least one is discriminating. However, an advantage of our tests is that it does not require the information the agents have to be identical across observations and, therefore, across sub-samples. Second, even if the complete sample $(x_{i,t})_{(i,t)\in\mI\times\mT}$ is rationalizable, we cannot reject the null hypothesis of statistical discrimination. This is because it is possible that agents have common beliefs because they are \textit{all} statistically discriminating against the disadvantaged group.\medskip

We now consider the case of \textit{taste-based} discrimination. Once again, consider the two sub-samples, where we condition on $[\theta_t=\theta]$ and $[\theta_t=\theta']$. If the two sub-samples are rationalizable, we can infer the ordering over thresholds in the two sub-samples and compare them. If the two orderings differ, there exists (at least) a pair of agents $(i,j)$ such that 
$\overline{u}_i^{\theta} \geq \overline{u}_j^{\theta}$ and $\overline{u}_j^{\theta'} > \overline{u}_i^{\theta'}$. This implies that if we assume that the two agents have common beliefs in each sub-sample---an assumption we cannot reject since the two sub-samples are rationalizable---then agent $j$ more often approves applications with characteristic $\theta$ than agent $i$, while the converse is true when the characteristic is $\theta'$. In other words, this is evidence of \textit{taste-based} discrimination. Finally, note that, if the two orderings coincide, we cannot reject the null hypothesis of no taste-based discrimination because this discrimination, while present, may not be reflected in the choice data.

\subsection{Committees} \label{sec:app_committees}

While there is a large theoretical literature on committee voting, there is substantially less structural empirical work. \citet{iaryczower2012} estimate a model of committee voting by using the judicial decisions in the US Supreme Court. Here, both states (whether the meaning of the law favors the plaintiff or defendant) and alternatives (the ruling) are binary. \citet{hansen2016} study individual voting behavior in Bank of England's Monetary Policy Committee. They model the decision making environment using a binary state (more/less inflationary state), binary alternative (vote for either high or low interest rates) and assume that the state is drawn independently from an identical distribution in each period.\medskip

Our model is quite similar to these papers albeit with two key distinctions. They allow for private information but do not allow for learning (\citealp{hansen2016} instead focus on the role of signaling) and they effectively assume that (conditioning on observable covariates) each decision problem that a committee faces is ex ante identical (prior to the observation of signals). Moreover, the aim in these papers is to estimate the underlying model parameters and for this they require additional functional form assumptions. By contrast, although our model allows us to determine some ordinal information (such as the ranking of voters' cutoffs), we cannot estimate any underlying parameters because they are not identified.\footnote{In this sense, our paper is closer to \citet{iaryczower2018} who develop a model of communication within the committee and estimate the effects of deliberation on the decisions of US appellate courts. Their model has multiple equilibria and is partially identified.}\medskip

One feature that is common to all these papers is that they assume committee members have the same prior beliefs in each decision problem before receiving their private signals. In other words, the heterogeneity across agents is in terms of their preferences or ``ideology'' instead of their prior beliefs or ``bias.''\footnote{\citet{iaryczower2012} discuss in footnotes 11 and 18 why allowing for heterogeneity in the prior beliefs lead to poorly behaved estimates.} A natural question is whether these two different sources of (single parameter) heterogeneity lead to distinct testable implications? This would be informative because we might think the sources of bias and ideology differ and so different interventions might be required should we aim to correct them.\medskip

\cref{thm:comon_inv} showed that, if we allow for co-monotone signals, then there is no additional generality in assuming different preferences for each agent. So instead, suppose we assume that agents have the same preferences, information in each period arrives via common signals (as in \cref{def:rat}) but agents have different prior beliefs. This effectively amounts to assuming that the committee is deliberative but that members differ in terms of bias as opposed to ideology. We can show that heterogeneity in prior beliefs or preferences have the same testable implications when the state is time-invariant but the former has strictly weaker (but nontrivial) restrictions when the state is time varying.\footnote{Formal statements and proofs available upon request.} In other words, differences between these sources of heterogeneity are testable.

\subsection{Counterfactual analysis}

A natural question is whether we can use the observed choice data $\mX$ to make counterfactual predictions about actions in a hypothetical period $T+1$. To fix ideas, consider the baseline model. \cref{thm:baseline} implies that we can rule out the subset of period $T+1$ actions that would generate a cycle. However, a consequence is that we can never rule out the outcome where all agents  pick the same alternative, that is, $x_{i,T+1}=x_{j,T+1}$ for all $(i,j)$. In the context of committee voting, this implies that a unanimous vote for either alternative is always possible in the counterfactual scenario.\medskip

That said, our model allows for considerably sharper \textit{conditional} predictions. To be more precise, if we knew the alternative that one agent plans to choose in period $T+1$, this can impose a lot more structure on the actions of the other agents. For instance, if a member $i$ of a committee announces that she plans to vote $x$ in period $T+1$, we might be able to rule out the event where the alternative $y$ gets a majority of votes and therefore is enacted. We end the paper by demonstrating how these observations can be employed in a simple example.

\begin{example}\label{eg:prediction}
	Consider the choice data $\mX$ of three agents $i$, $j$, $k$ and two time periods represented in the following table:
	\begin{center}
		\begin{tabular}{|c|c|c|c|}\hline
			& $i$ & $j$  & $k$   \\ \hline
			$t_1$ & $x$ & $y$ & $y$   \\ \hline
			$t_2$ & $x$ & $x$ & $y$ \\ \hline
		\end{tabular}
	\end{center}		
	
	\medskip
	
	Since there are no cycles, this data is rationalizable in the sense of \cref{def:rat}. Now suppose we want to predict the possible action profiles in a third period $t_3$. The following table describe all possible counterfactual action profiles (that do not generate cycles):\\
	
	\begin{center}		
		\begin{tabular}{|c|c|c|c|}\hline
			& $i$ & $j$  & $k$   \\ \hline
			$t_1$ & $x$ & $y$ & $y$   \\ \hline
			$t_2$ & $x$ & $x$ & $y$ \\ \hline
			$t_3$ & $x$ & $x$ & $x$ \\ \hline
		\end{tabular}
		\qquad\qquad
		\begin{tabular}{|c|c|c|c|}\hline
			& $i$ & $j$  & $k$   \\ \hline
			$t_1$ & $x$ & $y$ & $y$   \\ \hline
			$t_2$ & $x$ & $x$ & $y$ \\ \hline
			$t_3$ & $x$ & $x$ & $y$ \\ \hline
		\end{tabular}
		\qquad\qquad
		\begin{tabular}{|c|c|c|c|}\hline
			& $i$ & $j$  & $k$   \\ \hline
			$t_1$ & $x$ & $y$ & $y$   \\ \hline
			$t_2$ & $x$ & $x$ & $y$ \\ \hline
			$t_3$ & $x$ & $y$ & $y$ \\ \hline
		\end{tabular}
		\qquad\qquad
		\begin{tabular}{|c|c|c|c|}\hline
			& $i$ & $j$  & $k$   \\ \hline
			$t_1$ & $x$ & $y$ & $y$   \\ \hline
			$t_2$ & $x$ & $x$ & $y$ \\ \hline
			$t_3$ & $y$ & $y$ & $y$ \\ \hline
		\end{tabular}
	\end{center}

\bigskip

There are eight possible action profiles that are possible in period $t_3$ and the observed choice data eliminates half of those. Note that the example is such that the choices in the first two time periods completely order the cutoffs of all three agents ($\overline{u}_i<\overline{u}_j<\overline{u}_k$). Therefore, observing the choices from any additional time periods would not change the counterfactual predictions and we would still only be able to eliminate half the possible action profiles. 

\medskip

Now suppose that we knew that either agent $i$ or $j$ will choose $y$ in period $t_3$. We can immediately conclude that only the right two action profiles are possible and so $x$ cannot be chosen by a majority. In other words, knowing the choices of a single agent in a counterfactual time period can allow us to considerably sharpen our predictions. Finally, the insights from this example can easily be generalized to demonstrate how counterfactual analysis can be conducted when there are general preferences and/or multiple actions and/or co-monotone experiments.
\end{example}

\section{Concluding remarks}\label{sec:discussion}

We develop a simple yet rich framework via which we can test different assumptions on preferences and the underlying learning environment with very parsimonious data. The theoretical results can be used to uncover evidence for different types of discrimination and to detect differences in bias versus ideology in committee decision making. The model allows for counterfactual analysis that can be considerably sharpened if the actions of a subset of agents is known. Finally, we can test the assumption of common beliefs against other models of private learning (because they have distinct testable implications). This shows that it is possible to distinguish between different learning behaviors. This is a potentially useful result for the large behavioral economics literature that studies non-Bayesian learning.

\newpage

\appendix

\section{Proof of \cref{thm:mult_actions}}\label{app:mult_actions}

\begin{proof} ($\Rightarrow.$) In the text. \medskip 

($\Leftarrow.$) We first define an incomplete binary relation $\succ$ on pairs $(i,n)$ of agents $i\in\mI$ and indices of actions $n\in\{1,\dots,N\}$ as follows. For all $i\in \mI$, for all $n\in\{2,\dots,N\}$, \[(i,n)\succ (i,n-1).\] 
For all $i\neq j$, for all $n\in\{2,\dots,N\}$, 
\[(i,n)\succ (j,n'), \]
if there exists  $t\in\mT$ such that  $x_{i,t}=y_{n-1}$ and $x_{j,t}=y_{n'}$. All other pairs are not comparable. \medskip

This  relation reflects the order of the preference cutoffs. The first line captures the order $\overline{u}_{i,n}>\overline{u}_{i,n-1}$ of actions. As for  the second line, observe that, when $x_{i,t}=y_{n-1}$,  $x_{j,t}=y_{n'}$, if the data is rationalizable, then $\overline{u}_{i,n}>p_t>\overline{u}_{i,n-1}$ and $\overline{u}_{j,n'+1}>p_t>\overline{u}_{j,n'}$ which imply $\overline{u}_{i,n}>\overline{u}_{j,n'}$.\medskip
	
	We now argue that the relation $\succ$ does not have a cycle of any length, that is, there do not exist sequences  $((i_1,n_1),\dots,(i_k,n_k))$ for some $k\geq 2$ such that $$(i_1,n_1)\succ (i_2,n_2)\succ\cdots \succ (i_k,n_k)\succ (i_1,n_1).$$ 
	
	By contradiction, suppose that such a sequence exists. By definition of $\succ$, $(i_1,n_1)\succ (i_2,n_2)$ implies that there is a time period $t_2$ at which $x_{i_1,t_2}=y_{m_2}$ where $m_2\leq n_1-1$ and $x_{i_2,t_2}=y_{n_2}$ (if $i_2=i_1$ then $x_{i_1,t_2}=x_{i_2,t_2}=y_{m_2}=y_{n_2}$ with $m_2=n_2=n_1-1$).\medskip
	
	Similarly, $(i_2,n_2)\succ (i_3,n_3)$ implies that there is a time period $t_3$ at which $x_{i_2,t_3}=y_{\ell}$ where $\ell\leq n_2-1$ and $x_{i_3,t_3}=y_{n_3}$. Now let us compare periods $t_2$ and $t_3$. Observe that agent $i_2$ chooses a strictly lower action in $t_3$ ($y_{\ell}$ as opposed to $y_{n_2}$) and therefore the no cycles condition implies that agent $i$ must also choose a weakly lower action $x_{i_1,t_3}=y_{m_3}$ where $m_3\leq m_2 < n_1$. \medskip 
	
	Iterating the argument, there exists a period $t_k$ where agent $i_k$ chooses $x_{i_k,t_k}=y_{n_k}$ and agent $i$ chooses action $x_{i_1,t_k}=y_{m_k}$ where $m_k < n_1$.\medskip
	
	Now observe that if $i_k=i_1$, then $x_{i_k,t_k}=y_{n_k}=x_{i_1,t_k}=y_{m_k}$ which is a contradiction because $(i_k,n_k)\succ (i_1,n_1)$ means that $n_k=m_k=n_1+1$. \medskip
	
	Therefore, it must be the case that $i_k\neq i_1$. Consequently, there exists period $t_1$ in which $x_{i_k,t_1}=y_{n_k-1}$ and $x_{i_1,t_1}=y_{n_1}$. Compare periods $t_k$ and $t_1$. Agent $i_k$ chooses a strictly lower action ($y_{n_k-1}$ versus $y_{n_k}$) in $t_1$ whereas agent $i_1$ chooses a strictly higher action ($y_{n_1}$ versus $y_{m_k}$). These choices constitute a cycle, a contradiction.
	
	\medskip
	
	In addition, we define the symmetric part $\sim$ as $(i,n) \sim (i,n)$ for all $(i,n)$. Thus, the relation $\succ \cup \sim$ is reflexive and consistent. By Theorem 3 (that generalizes Szpilrajn's Theorem) in \citet{suzumura1976}, there exists a complete and transitive binary relation $\succcurlyeq^*$ on the set of pairs $(i,n)$  that extends $\succ \cup \sim$. In turn, this implies that we can find numbers $\overline{u}_{i,n}$ for all $(i,n)$ that satisfy $\overline{u}_{i,n}>\overline{u}_{j,n'}$ when $(i,n)\succcurlyeq^* (j,n')$. In particular, note that
	\begin{enumerate}[label=\roman*.]
		\item $\overline{u}_{i,n}>\overline{u}_{i,n'}$ for $n>n'$ and
		\item $\overline{u}_{i,n}>\overline{u}_{j,n'}$ if there is a time period $t\in\mT$ such that $x_{i,t}=y_{n-1}$ and $x_{j,t}=y_{n'}$.
	\end{enumerate}

	\medskip
	
	Item  i gives us the correctly ordered cutoffs for each agent, so these cutoffs generate valid preferences. Item ii implies that, for every $t\in \mT$, we can find a $p_t$ such that $\overline{u}_{i,n+1}>p_t>\overline{u}_{i,n}$ where $x_{i,t}=y_n$ for all $i\in \mI$. To see this, note that, by point ii above, $x_{i,t}=y_n$ and $x_{j,t}=y_{n'}$ implies $\overline{u}_{i,n+1}>\overline{u}_{j,n'}$. Therefore,
	\[\min_{i,n}\{\overline{u}_{i,n+1}\; | x_{i,t}=y_n \}>\max_{j,n'}\{\overline{u}_{j,n'}\; |  x_{j,t}=y_{n'} \},\]
	and any $p_t$ between the left and right sides will suffice.
	
	\medskip
	
To complete the proof, we can construct the signals and experiments by following the same steps as in the proof of \cref{thm:baseline} from the main text. 
\end{proof}

\section{ From marginals to joint distributions }\label{app:private_signals}
\begin{lemma}\label{app:lemma_expe} Suppose that there exist positive numbers $(\pi_{i}^x,\pi_{i}^y)_{i \in \mathcal{I}}$ such that $\sum_i \pi_i^x < 1$, $\sum_{i} \pi_i^y < 1$, and
\[[\pi_i^x - \pi_i^y][\pi_j^x-\pi_j^y] >0, \]
for all $(i,j)$. Then, there exists a strictly co-monotone experiment $\pi: \{\omega^x,\omega^y\} \rightarrow \Delta(\times_{i \in \mI} S_i)$ and a profile of signals $s=(s_1,\dots,s_I)$ such that 
\[(\pi_{i}(s_i|\omega^x), \pi_{i}(s_i|\omega^y)) = 
(\pi_{i}^x,\pi_{i}^y),\]
for all $i$. Moreover, the profile of signal $s$ is in the support of the experiment. 
\end{lemma}

\begin{proof}
We construct $\pi(\cdot|\omega^x)$. 
For all $\varepsilon\geq 0$,  let $\pi_i^{x}(\varepsilon)$ be a solution to
\[(1-\varepsilon)\pi^x_{i}(\varepsilon)+ \varepsilon \frac{1}{I} = \pi^x_i.\]
Choose $\varepsilon$ small enough such that $\pi^x_i(\varepsilon) \in (0,1)$ and $\sum_i \pi^x_i(\varepsilon) <1$. Since $\pi^x_i(0)= \pi^x_i$, this is possible by continuity. \medskip 

Let $S_i=\{s^0,s^1,\dots,s^{I}\}$ for all $i$ and $S^*= \times_{i \in \mathcal{I}} (S_i \setminus \{s^0\})$. We define $\pi(\cdot|\omega^x)$ as follows: 
\begin{equation*}
\pi(s_1,\dots,s_n|\omega^x)=
\begin{cases}
(1-\varepsilon) \pi_i^x(\varepsilon) + \varepsilon \frac{1}{I^I} & \text{if\;} s_1=\dots =s_n = s^i \\
\varepsilon \frac{1}{I^I} & \text{if\;} s  \in S^* \setminus \bigcup_{i}\{(s^i,\dots,s^i)\}\\
(1-\varepsilon) (1-\sum_i \pi_i^x(\varepsilon) )& \text{if\;} s_1=\dots=s_n=s^0 \\
0 & \text{otherwise}. 
\end{cases}
\end{equation*}
 It is routine to check that $\pi(\cdot|\omega^x)$ is indeed a probability over $\times_{i \in \mathcal{I}}S_i$. Note that the support of $\pi(\cdot|\omega^x)$ is $S^* \cup \{(s^0,\dots,s^0)\}$. For all $s \in S^*$,  agent $i$'s marginal  $\pi_i(s_i|\omega^x)$ is 
\[(1-\varepsilon)\pi^x_{\ell}(\varepsilon)+ \varepsilon \frac{1}{I} = \pi^x_{\ell} \]
if $s_i=s^{\ell}$ with $\ell \in \{1,\dots,I\}$.\medskip

The probability  $\pi(\cdot|\omega^y)$ is constructed analogously. Note that the epsilons we use to construct $\pi(\cdot|\omega^x)$ and $\pi(\cdot|\omega^y)$ may differ. However, for any epsilon smaller than the minimum of the two of them, we continue to have two well-defined probabilities  with the same epsilon.\medskip

It remains to check that the experiment is co-monotone. Consider any profile of signals $s \in S^*$. We have that  
\begin{equation*}
[\pi_i(s_i|\omega^x) - \pi_i(s_i|\omega^y)] [\pi_j(s_j|\omega^x) - \pi_j(s_j|\omega^y) ] = 
[\pi_{\ell}^x - \pi_{\ell}^y ] [[\pi_{\ell'}^x - \pi_{\ell'}^y ] >0,
\end{equation*}
with $(s_i,s_j)= (s^{\ell},s^{\ell'})$, $(\ell,\ell') \in \mathcal{I} \times \mathcal{I}$.

Finally, if the signal profile is $(s^0,\dots,s^0)$, we have that 
\begin{equation*}
[\pi_i(s^0|\omega^x) - \pi_i(s^0|\omega^y)] [\pi_j(s^0|\omega^x) - \pi_j(s^0|\omega^y) ] = [\pi_i(s^0|\omega^x) - \pi_i(s^0|\omega^y)]^2 > 0,
\end{equation*}
for all pairs $(i,j)$. This follows immediately from the fact that all agents have the same marginal over $s^0$, regardless of the state.\medskip

The experiment is therefore co-monotone. This completes the proof. 
\end{proof}

\medskip
We can illustrate the construction when there are two agents. When the state is $\omega^x$, the joint distribution over signals is given by

\begin{table}[h]
\begin{tabular}{c|ccc}
 & $s^1$ & $s^2$ & $s^0$ \\ \hline
$s^1$  & $(1-\varepsilon)\pi_1^x(\varepsilon) + \varepsilon \frac{1}{4}$ &  $\varepsilon \frac{1}{4}$ & $0$ \\
$s^2$  & $\varepsilon \frac{1}{4}$ & $(1-\varepsilon)\pi_2^x(\varepsilon) + \varepsilon \frac{1}{4}$ & $0$ \\
$s^0$ & $0$ & $0$ & $(1-\varepsilon)(1-\pi_1^x(\varepsilon)-\pi_2^x(\varepsilon))$\\
\end{tabular}
\end{table}
Thus, conditional on any signal, the two agents have the same marginals. Moreover, we construct the numbers $\epsilon$, $\pi_1^x(\varepsilon)$ and $\pi_2^x(\varepsilon)$ such that 
\[(1-\varepsilon)\pi_i^x(\varepsilon) + \varepsilon \frac{1}{4}=\pi_i^x\]
for all $i$, and the probability is well-defined.\medskip

We do a similar construction when the state is $\omega^y$.  It is then immediate to verify that the experiment is co-monotone. Finally, if agent $i$ receives the signal $s^i$, the ratio of marginals is 
\[\frac{\pi_i^x}{\pi_i^y}, \]
as required in the main text. 
\newpage 
\bibliographystyle{ecta}
\bibliography{DynamicChoices_Refs}

\end{document}